\documentclass[fleqn,12pt,twoside]{article}
\usepackage{tipa}

\usepackage[headings]{espcrc1}
\readRCS
$Id: espcrc1.tex,v 1.2 2004/02/24 11:22:11 spepping Exp $
\usepackage{graphicx}
\usepackage[figuresright]{rotating}

\newtheorem{theorem}{Theorem}

\newtheorem{conjecture}{Conjecture}

\newtheorem{lemma}{Lemma}

\newtheorem{problem}{Problem}
\newtheorem{proposition}{Proposition}

\newenvironment{proof}[1][Proof.]{\begin{trivlist}
\item[\hskip \labelsep {\bfseries #1}]}{\end{trivlist}}

\newenvironment{acknowledgement}[1][Acknowledgement]{\begin{trivlist}
\item[\hskip \labelsep {\bfseries #1}]}{\end{trivlist}}

\newcommand{\AmS}{{\protect\the\textfont2
  A\kern-.1667em\lower.5ex\hbox{M}\kern-.125emS}}

\usepackage{amsmath}
\usepackage{amsfonts}
\title{A remark on Petersen coloring conjecture of Jaeger}

\author{Vahan V. Mkrtchyan\address[MCSD]{Department of Informatics and Applied Mathematics,\\
Yerevan State University, Yerevan, 0025, Armenia}%
\address{Institute for Informatics and Automation Problems,\\
National Academy of Sciences of Republic of Armenia, 0014, Armenia}
\thanks{email: vahanmkrtchyan2002@\{ysu.am, ipia.sci.am,
yahoo.com\}}}

\runtitle{A remark on Petersen coloring conjecture of Jaeger} \runauthor{Vahan
Mkrtchyan}

\begin{document}

% typeset front matter
\maketitle

\begin{abstract}
If $G$ and $H$ are two cubic graphs, then we write $H\prec G$, if $G$ admits a proper edge-coloring $f$ with edges of $H$, such that for each vertex $x$ of $G$, there is a vertex $y$ of $H$ with $f(\partial_G(x))=\partial_H(y)$. Let $P$ and $S$ be the Petersen graph and the Sylvester graph, respectively. In this paper, we introduce the Sylvester coloring conjecture. Moreover, we show that if $G$ is a connected bridgeless cubic graph with $G\prec P$, then $G=P$. Finally, if $G$ is a connected cubic graph with $G\prec S$, then $G=S$.
\end{abstract}

\section{Introduction}

The graphs considered here are finite and undirected. They do not contain loops though they may contain multiple edges.  For a vertex $v$ of $G$ let $\partial_G(v)$ be the set of edges of $G$ incident to $v$.

Let $G$ and $H$ be two cubic graphs. Then an $H$-coloring of $G$ is a proper edge-coloring $f$ with edges of $H$, such that for each vertex $x$ of $G$, there is a vertex $y$ of $H$ with $f(\partial_G(x))=\partial_H(y)$. If $G$ admits an $H$-coloring, then we will write $H\prec G$.

If $H\prec G$ and $f$ is an $H$-coloring of $G$, then for any adjacent edges $e,e'$ of $G$, the edges $f(e),f(e')$ of $H$ are adjacent. Moreover, if the graph $H$ contains no triangle, then the converse is also true, that is, if a mapping $f: E(G)\rightarrow E(H)$ has a property that for any two adjacent edges $e$ and $e'$ of $G$, the edges $f(e)$ and $f(e')$ of $H$ are adjacent, then $f$ is a $H$-coloring of $G$.

Let $P$ be the well-known Petersen graph (figure \ref{Petersen}) and let $S$ be the Sylvester graph (figure \ref{Sylvester}). Both of them have ten vertices. The Petersen coloring conjecture of Jaeger states:

\begin{conjecture} (Jaeger, 1988 \cite{Jaeger}) For each bridgeless cubic graph $G$, one has $P\prec G$.
\end{conjecture} 

\begin{figure}[h]
\begin{center}
\includegraphics [height=23pc]{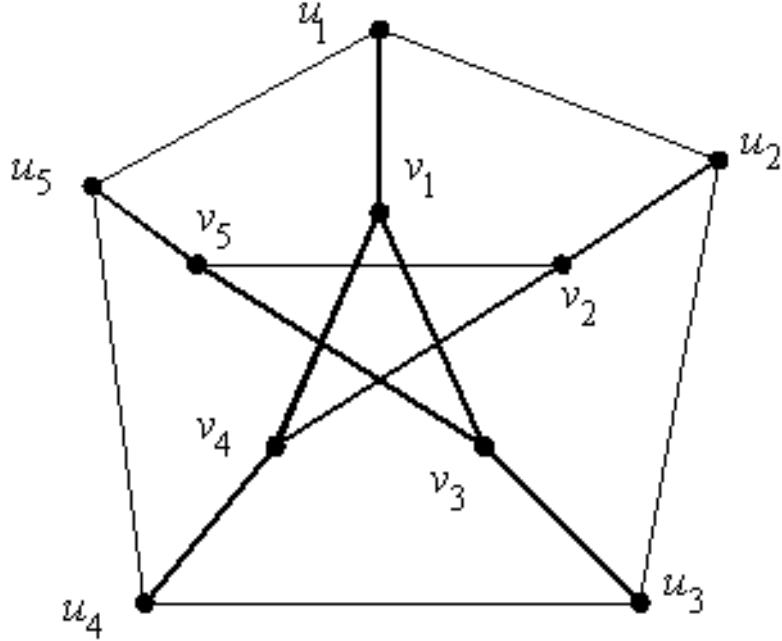}\\
\caption{The Petersen graph}\label{Petersen}
\end{center}
\end{figure}

\begin{figure}[h]
\begin{center}
\includegraphics [height=23pc]{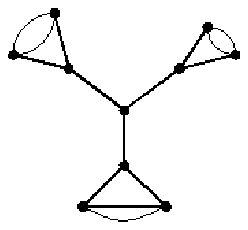}\\
\caption{The Sylvester graph}\label{Sylvester}
\end{center}
\end{figure}

The conjecture is difficult to prove, since it can be easily seen that it implies the following two classical conjectures:

\begin{conjecture} (Berge-Fulkerson, 1972 \cite{BergeFulkerson}) Any bridgeless cubic graph $G$ contains six (not necessarily distinct) perfect matchings $F_1,...,F_6$ such that any edge belongs to exactly two of them.
\end{conjecture}

\begin{conjecture} ($(5,2)$-cycle-cover conjecture, \cite{SeymourCDC,SzekeresCDC}) Any bridgeless graph $G$ (not necessarily cubic) contains five even subgraphs such that any edge belongs to exactly two of them.
\end{conjecture} Recall that a subgraph $H$ of a graph $G$ is even, if any vertex $x$ of $G$ has even degree in $H$.

Related with Jaeger conjecture, we would like to introduce the following

\begin{conjecture} For each cubic graph $G$, one has $S\prec G$.
\end{conjecture}In the direct analogy with Jaeger conjecture, we call this Sylvester coloring conjecture.

One may wonder whether there are other ($\neq P$) bridgeless cubic graphs $H$, such that for any bridgeless cubic graph $G$ one has $H\prec G$? Similarly, we can look for other ($\neq S$)  cubic graphs $H$, such that for any cubic graph $G$ one has $H\prec G$. It is easy to see that there are infinitely many disconnected bridgeless cubic graphs $H$ meeting this condition provided that Jaeger conjecture is true (hint: take any disconnected bridgeless cubic graph $H$, which contains a connected component that is isomophic to the Petersen graph). A similar construction works with Sylvester graph. Thus it is natural to re-state these questions as follows:

\begin{problem} Are there other ($\neq P$) connected bridgeless cubic graphs $H$, such that for any bridgeless cubic graph $G$ one has $H\prec G$ provided that Jaeger conjecture is true?
\end{problem}

\begin{problem} Are there other ($\neq S$) connected cubic graphs $H$, such that for any cubic graph $G$ one has $H\prec G$ provided that Sylvester coloring conjecture is true?
\end{problem}

It is easy to see that the theorems \ref{MainResult} and \ref{MainResult2} proved below imply that the answers to these problems are negative.

Non-defined terms and concepts can be found in \cite{BondyMurty}.

\section{The main results}

For the proof of the main results, we will need the following:

\begin{proposition}\label{proposPetersen} Let $G$ be a non 3-edge-colorable bridgeless cubic graph that has at most ten vertices. Then $G=P$.
\end{proposition}

\begin{proposition}\label{proposSylvester} Let $G$ be a cubic graph that has no a perfect matching and has at most ten vertices. Then $G=S$.
\end{proposition}

\begin{lemma}\label{PrecLemma} Suppose that $G$ and $H$ are cubic graphs with $H\prec G$, and let $f$ be an $H$-coloring of $G$. Then:
\begin{itemize}
	\item [(A)] If $M$ is any matching of $H$, then $f^{-1}(M)$ is a matching of $G$;
	\item [(B)] $\chi'(G)\leq \chi'(H)$, where $\chi'(G)$ is the chromatic index of $G$;
	\item [(C)] If $M$ is any perfect matching of $H$, then $f^{-1}(M)$ is a perfect matching of $G$.	
\end{itemize}
\end{lemma}

\begin{proof}(A) is trivial. 

(B) Let $\chi'(H)=l$. Then 
\begin{equation*}
E(H)=M_1\cup... \cup M_{l},
\end{equation*}where for $i=1,...,l$ $M_i$ is a matching. This implies that:
\begin{equation*}
E(G)=f^{-1}(M_1)\cup... \cup f^{-1}(M_l).
\end{equation*} Now, by (A), we have that for $i=1,...,l$ $f^{-1}(M_i)$ is a matching. Thus, $G$ is $l$-edge-colorable.

(C) By (A), we have that $f^{-1}(M)$ is a matching. Let $v$ be any vertex of $G$. Since $M$ is a perfect matching of $H$, we have $f(\partial_G(v))\cap M\neq \emptyset$. Thus $f^{-1}(M)$ is a perfect matching. $\square$
\end{proof}

We are ready to prove:

\begin{theorem}\label{MainResult} If $G$ is a connected bridgeless cubic graph with $G\prec P$, then $G=P$.
\end{theorem}

\begin{proof} By (B) of lemma \ref{PrecLemma} $G$ is non 3-edge-colorable. Let $f$ be a $G$-coloring of $P$. If $e\in E(G)$, then we will say that $e$ is used (with respect to $f$), if $f^{-1}(e)\neq \emptyset$.

First of all, let us show that if an edge $e$ of $G$ is used, then any edge adjacent to $e$, is also used.

So let $e=uv$ be a used edge of $G$. For the sake of contradiction, assume that $v$ is incident to an edge $z\in E(G)$ that is not used. Suppose that $\partial_G(u)=\{a,b,e\}$. Observe that $a$ and $b$ are also used.

Due to symmetry of Petersen graph, we can assume that $f(u_3u_4)=e$. Suppose that $f(u_4u_5)=a$ and $f(u_4v_4)=b$ (figure \ref{Petersen}). Since the edge $z$ is not used, we have: $f(\partial_P(u_3))=\partial_G(u)=\{a,b,e\}$. Again, due to symmetry of Petersen graph, we can assume that $f(u_3v_3)=b$ and $f(u_2u_3)=a$.

Let $a_1=f(u_1u_5)$, $a_2=f(u_1u_2)$. Observe that since $f$ is a $G$-coloring of $P$, we have that $a_1$ and $a_2$ are adjacent edges of $G$. Moreover, each of them is adjacent to $a$. Similarly, the edges $b_1=f(v_1v_4)$ and $b_2=f(v_1v_3)$ of $G$ are adjacent, and each of them is adjacent to $b$.

We will differ three cases:\\

Case 1: The edges $a_1$, $a_2$ and $a$ do not form a triangle in $G$.\\

Observe that in this case $f(u_1v_1)=a$. This implies that the edges $a$, $b_1$, $b_2$  must be incident to the same vertex. However, this is possible only when $b_1$ and $b_2$ are two parallel edges connecting the other ($\neq u$) end-vertices of edges $a$ and $b$, which is a contradiction, since $e$ cannot be a bridge.\\

Case 2: The edges $b_1$, $b_2$ and $b$ do not form a triangle in $G$.\\

This case is similar to case 1.\\

Case 3: The edges $a_1$, $a_2$ and $a$ form a triangle in $G$. Similarly, $b_1$, $b_2$ and $b$ form a triangle. \\

Let $a_3$ be the edge of $G$ that is adjacent to $a_1$, $a_2$ and is not adjacent to $a$. Note that such an edge exists since $G$ is bridgeless. Similarly, let $b_3$ be the edge of $G$ that is adjacent to $b_1$, $b_2$ and is not adjacent to $b$. Observe that $a_3=f(u_1,v_1)=b_3$, and hence $a_3=b_3$. On the other hand, since $G$ is bridgeless and $a\neq b$, we have $a_3\neq b_3$, which is a contradiction.\\

We are ready to complete the proof of theorem \ref{MainResult}. Observe that since $G$ is connected, we have that all edges of $G$ are used, and hence $|E(G)|\leq |E(P)|$, or $|V(G)|\leq |V(P)|=10$. Proposition \ref{proposPetersen} implies that $G=P$.

$\square$
\end{proof}

\begin{theorem}\label{MainResult2} Let $G$ be any connected cubic graph with $G\prec S$. Then $G=S$.
\end{theorem}

\begin{proof} Let $G$ be a connected cubic graph with $G\prec S$, and let $f$ be the corresponding coloring. Clearly, $G$ has no a perfect matching (see (C) of lemma \ref{PrecLemma}). 

Again, an edge $e\in E(G)$ is used (with respect to $f$) if $f^{-1}(e)\neq \emptyset$.

First of all let us show that if an edge of $G$ is used, then all edges adjacent to it are used, too. Suppose that $a=uv$ is a used edge of $G$ that is adjacent to a non-used edge. Suppose that $v$ is incident to a non-used edge. Let $b$ and $c$ be the other edges incident to $u$.

Since $a$ is used, there is $e\in E(S)$ with $f(e)=a$. We will consider three cases.\\

Case 1: The multiplicity of $e$ is two in $S$. Suppose that the edge parallel to $e$ is colored by $b$. Then the two edges of $S$ forming a triangle with $e$ must be colored with $c$ which is impossible.\\

Case 2: $e$ is adjacent to an edge of multiplicity two in $S$. Then the two parallel edges must be colored with $b$ and $c$, and hence the other edge that is adjacent to the same two parallel edges must be colored with $a$, which is again a contradiction.\\

Case 3: $e$ is a bridge in $S$. Let $e'$ and $e''$ be two non-bridge edges of $S$ that are adjacent to $e$. We can assume that $f(e')=b$ and $f(e'')=c$. Finally, let $g$ and $h$ be the two parallel edges adjacent to $e'$ and $e''$. By Case 1, $g$ and $h$ cannot be colored by $a$, hence the colors of $g$ and $h$ must be adjacent to both $b$ and $c$. It is clear that if $x$ and $y$ denote the colors of $g$ and $h$, then $x$ and $y$ must form a multi-edge in $G$.

Let $z$ be the bridge of $S$ colored by $b$. By Case 2, none of the edges adjacent to $z$ and non-adjacent to $e$, can be colored by $a$, hence these two edges are colored $x$ and $y$. Now, observe that the two parallel edges adjacent to these two edges must be colored by $b$ and $y$, or $c$ and $y$, which is a contradiction.\\

To complete, the proof, let us note that since $G$ is connected, the proved property implies that all edges of $G$ are used, hence $|E(G)|\leq |E(S)|=15$ or $|V(G)|\leq |V(S)|=10$. Proposition \ref{proposSylvester} implies that $G=S$.
$\square$

\end{proof}

\begin{acknowledgement}
The author would like to thank Giuseppe Mazzuoccolo for useful discussions over the subject. He is also thankful to one of the referees who has significantly shortened the proof of theorem \ref{MainResult}.
\end{acknowledgement}

\end{document}